\documentclass[runningheads,envcountsame]{llncs}
\usepackage{standalone}
\usepackage{graphicx} 
\usepackage{xspace}
\usepackage{amsmath}
\usepackage{amssymb}
\usepackage{enumitem}
\usepackage{nicefrac}
\usepackage[numbers,sort&compress]{natbib}
\usepackage{hyperref}
\usepackage{cleveref}
\usepackage{tikz}
\usepackage{thmtools}
\usepackage{hyperref}
\usepackage{cleveref}
\usepackage{apxproof}

\usetikzlibrary{arrows.meta, positioning, fit, backgrounds, calc}

\usepackage{xcolor}

\title{Interval number for tournaments \texorpdfstring{\\}{} in \texorpdfstring{$P_3$}{P3}-convexity}
\author{Idian C. Capozzoli\inst{1}\orcidID{0000-0002-8282-0701}, Yan S. Couto\inst{1}\orcidID{0009-0005-5850-8696} \and Enrique Junchaya\inst{1}\orcidID{0009-0003-1898-2102}}
\authorrunning{I.\,C. Capozzoli, Y.\,S. Couto and E.\,Junchaya}
\institute{University of São Paulo, São Paulo, SP, Brazil\\
\email{\{idian,yancouto,enriquejh\}@ime.usp.br}}

\usepackage{verbatim}
\usepackage{soul} 
\soulregister\cite7
\usepackage[commandnameprefix=ifneeded,authormarkup=none]{changes}
\definechangesauthor[color=orange]{yan}
\definechangesauthor[color=red]{idian}
\definechangesauthor[color=blue]{enrique}
\sethighlightmarkup{\IfIsColored{{\sethlcolor{authorcolor!30}\hl{#1}}}{#1}}

\newcommand{\Oh}{\mathcal{O}} 
\renewcommand{\mid}{\colon}

\newcommand{\cclass}[1]{\textbf{\textup{#1}}\xspace} 
\newcommand{\NP}{\cclass{NP}}
\renewcommand{\P}{\cclass{P}}
\newcommand{\W}[1]{\cclass{W[#1]}}
\newcommand{\FPT}{\cclass{FPT}}

\newcommand{\ETH}{\textbf{ETH}\xspace}

\usepackage[d]{esvect}
\renewcommand{\overrightarrow}[1]{\vv{#1}}

\DeclareMathOperator{\hn}{hn}

\DeclareMathOperator{\itn}{in}

\DeclareMathOperator{\gc}{g}
\DeclareMathOperator{\pc}{P_3}
\DeclareMathOperator{\psc}{P_3^*}
\DeclareMathOperator{\ogc}{\overrightarrow{g}}
\DeclareMathOperator{\opc}{\overrightarrow{P_3}}
\DeclareMathOperator{\opsc}{\overrightarrow{P_3^*}}
\DeclareMathOperator{\xc}{{\cal X}}
\DeclareMathOperator{\oxc}{\overrightarrow{\xc}}

\DeclareMathOperator{\ohn}{\overrightarrow{\hn}}

\DeclareMathOperator{\oin}{\overrightarrow{\itn}}

\DeclareMathOperator{\oinx}{\oin_{\xc}}

\DeclareMathOperator{\oing}{\oin_{\gc}}
\DeclareMathOperator{\ohng}{\ohn_{\gc}}

\DeclareMathOperator{\oinp}{\oin_{\pc}}

\DeclareMathOperator{\oinps}{\oin_{\psc}}
\DeclareMathOperator{\ohnps}{\ohn_{\psc}}

\DeclareMathOperator{\intfunc}{I}

\DeclareMathOperator{\oi}{\overrightarrow{\intfunc}}

\DeclareMathOperator{\oifp}{\oi_{\pc}}
\DeclareMathOperator{\oifps}{\oi_{\psc}}
\DeclareMathOperator{\oifx}{\oi_{\xc}}

\newcommand{\odn}{\gamma^+}

\usepackage{tabularx}
\usepackage{nameref}
\usepackage[skins,listings]{tcolorbox}
\newcommand{\pname}[1]{\textsc{#1}}
\newtcolorbox{ProblemBox}[2][]{
	colframe=black!20!white,
	coltitle=black,
	arc=1.5mm,
	boxrule=.15mm,
	colbacktitle=white,
	colback=white,
	enhanced,
	adjusted title=flush left,
	attach boxed title to top left={yshift=-3mm,xshift=3mm},
	boxed title style={colframe=white,righttitle=-3mm,lefttitle=-3mm},
	title=#2,#1
}
\newcommand{\problemDef}[4]{
	\begin{ProblemBox}[label={#4},nameref={#1}]{\pname{#1}}
		\begin{tabularx}{\textwidth}{l l} 
			\textnormal{Input:} & \quad \textnormal{#2} \\ 
			\textnormal{Question:} & \quad \textnormal{#3}
		\end{tabularx}
	\end{ProblemBox}
}

\newcommand{\paraProblemDef}[5]{
	\begin{ProblemBox}[label={#5},nameref={#1}]{\pname{#1}}
		\begin{tabularx}{\textwidth}{l l} 
			\textnormal{Input:} & \quad \textnormal{#2} \\ 
			\textnormal{Parameter:} & \quad \textnormal{#3}. \\ 
			\textnormal{Question:} & \quad \textnormal{#4}
		\end{tabularx}
	\end{ProblemBox}
}

\begin{document}

\maketitle

\begin{abstract}
    We study the complexity of determining the interval numbers of tournaments in the~$\opc$ and~$\opsc$ convexities, denoted by~$\oinp(T)$ and~$\oinps(T)$ on a tournament~$T$.

    For each~$\oxc \in \{\opc, \opsc\}$, we show that determining whether~\(\oinx(T) \leq k\) is~\W2-complete when parameterized by~$k$. Moreover, under \ETH, we show that there is no parameterized algorithm for that problem with running time~\(f(k)\, n^{o(k)}\) on an~$n$-vertex tournament, where~$f$ is any computable function.

    For the~$\opc$-convexity, we also show that~$\oinp(T) = \Oh(\log n)$, which yields a simple quasi-polynomial~$n^{\Oh(\log n)}$ brute-force algorithm. On the other hand, under \ETH, we show that the problem is \NP-intermediate, that is, it is neither \NP-hard nor in~\P.

    For the~\(\opsc\)-convexity, the same brute force algorithm is not quasi-polynomial, since we present a family of instances with~\(\oinps(T) = \Theta(n)\). We conjecture that this problem is~\NP-complete.
\end{abstract}

\section{Introduction}

Convexity has long been studied in different areas of mathematics, including in graphs, and even the class of tournaments~\cite{varlet_convexity_1976}.
We start by briefly introducing the relevant definitions of convexity in oriented graphs, which is our focus.

\subsection{Convexity and oriented graphs}

We use the graph notation of \citet{bondy_graph_2008}. A directed graph~$D$ consists of a set of vertices~$V(D)$ and arcs~$A(D) \subseteq V(D) \times V(D)$ between them, where an arc from~$u$ to~$v$ is denoted as~$uv$. An orientation is a directed graph with no opposite arcs ($uv$ and~$vu$), that is, it is equivalent to an undirected graph together with an orientation of its edges. A tournament is the orientation of a complete graph.
For an orientation $D$ and a vertex $v \in V(D)$, we denote by $N^+(v)$ the set of vertices of $D$ that are the heads of the outgoing arcs of $v$, and define~$N^-(v)$ analogously.

An \emph{interval convexity}~\(\oxc\) over an orientation~$D$ is given by a function~$\oifx: 2^{V(D)} \to 2^{V(D)}$ such that, for any~$U \subseteq V(D)$, we have that~$U \subseteq \oifx(U)$ and
\[
\oifx(U) = \begin{cases}
  U, & \text{if } |U| < 2, \\
  \bigcup\limits_{\{u,v\} \subseteq U}{\oifx(\{u,v\})},  & \text{if } |U| \geq 2.
\end{cases}
\]

Note that~$\oifx$ can be totally defined by the value it returns for pairs of vertices. We are interested in two specific types of convexity:
\begin{itemize}
    \item In \emph{two-path convexity}, or~\(\opc\)-convexity, the function~$\oifp$ returns, for each pair of vertices~$\{u,v\}$, all vertices in any path of length two from~$u$ to~$v$ or from~$v$ to~$u$;
    \item In \emph{distance-two convexity}, or~\(\opsc\)-convexity, the function~$\oifps$ returns, for each pair of vertices~$\{u,v\}$, all vertices in any \emph{shortest} path of length two from~$u$ to~$v$ or from~$v$ to~$u$.
\end{itemize}

For any convexity~$\oxc$, a set~$U$ is said to be \emph{convex} if~$\oifx(U) = U$, and we say~$U$ is an \(\oxc\)-\emph{interval set} if~$\oifx(U) = V(D)$. The interval number of~$D$ in relation to~\(\oxc\), denoted~\(\oinx(D)\), is the size of the smallest \(\oxc\)-interval set of~$D$.

\subsection{Our results and organization}

Our results focus on determining the interval number of a tournament~$T$, in convexities~\(\opc\) and~\(\opsc\).
We start in~\Cref{sec:prelim} by reviewing relevant results in complexity.

In~\Cref{sec:opc} we consider deciding whether~\(\oinp(T) \leq k\). We present a quasi-polynomial algorithm for that problem in~\Cref{subsec:oinp-alg}, which is based on the fact that a tournament always has a~\(\opc\)-interval set of logarithmic size. Then we consider the version of the problem parameterized by~$k$. In~\Cref{subsec:oinp-w2-hard,subsec:oinp-w2}, we show that this version is both~\W2-hard and in~\W2, concluding it is~\mbox{\W2-complete}, and also showing it is~\NP-intermediate assuming \ETH.

In~\Cref{sec:opsc} we consider deciding whether~\(\oinps(T) \leq k\). We start in~\Cref{subsec:oinps-n} by recording a reduction of this problem to strongly connected components.
Then, we present a family of strongly connected tournaments for which the minimum~\(\opsc\)-interval set has linear size. This result suggests that there might not be a quasi-polynomial algorithm for this problem. On the other hand, we consider the version of the problem parameterized by~$k$, and in~\Cref{subsec:oinps-w2-hard,subsec:oinps-w2} we show it is both~\W2-hard and in~\W2, concluding it is~\W2-complete.

We conclude in~\Cref{sec:conclusion}, examining further research areas.

\subsection{Related work}

Convexity in oriented graphs, and even in tournaments, has been studied since the 70s~\cite{erdos_remarks_1972,varlet_convexity_1976}. However, work on undirected graphs~\cite{everett_hull_1985} received more attention than its directed counterparts~\cite{chartrand_hull_2003} until more recently.

A \emph{hull set} is a set for which its convex hull\footnote{The convex hull of~$U$ is the smallest superset of~$U$ which is convex, or~$I^{|V(G)|}(U)$.} is the whole graph. The \emph{hull number} of a graph is the size of its smallest hull set, and the problem of determining the hull number has been studied in directed and undirected graphs~\cite{everett_hull_1985,chartrand_hull_2003}. Both hull and interval numbers have been studied in several subclasses of directed graphs~\cite{araujo_hull_2022,araujo_hull_2026,medeiros_convexidade_2024}, such as tournaments, underlying bipartite, split and cobipartite graphs.

In fact, it has been shown that determining~\(\oinp\) or~\(\oinps\) is~\NP-hard even when the underlying graph is chordal bipartite by~\citet{araujo_hull_2026}, who also cite as an open problem (Problem 32) whether one can compute~\(\oinp\) for a tournament in polynomial time. A PhD Thesis in Portuguese~\cite{medeiros_convexidade_2024} presents partial answers to that problem, by showing that for~\(\oinp\) \ETH implies it has no polynomial algorithms, and for~\(\oinps\) it is~\W2-hard. In comparison, for~\(\oinp\), we actually show it is \W2-complete. For~\(\oinps\), we present a much simpler proof of~\W2-hardness, as well as showing it is in~\W2.

\section{Preliminaries}\label{sec:prelim}

\subsection{Complexity}

Set \P contains decision problems which admit a polynomial time solution, while set~\NP contains decision problems which admit a \emph{nondeterministic} polynomial time solution. A problem~$\pi$ is \NP-hard if all problems in \NP have a polynomial time reduction to~\(\pi\), and it is~\NP-complete if~\(\pi \in \NP\) and~\(\pi\) is \NP-hard.

The classic problem \textsc{3-SAT} is \NP-complete, and the \emph{Exponential Time Hypothesis} (\ETH), stated by~\citet{impagliazzo_complexity_2001}, says that any algorithm for \textsc{3-SAT} must take~$2^{\Omega(n)}$ time, where~$n$ is the size of the input formula.

\subsection{Parameterized complexity}

While a problem may be~\NP-complete, it might be possible to design efficient solutions for it if the input is restricted, and that is the topic of parameterized complexity.
A parameterized problem is a pair~$(Q, k)$ where~$Q$ is a problem and~$k$ is its parameterization. Problem~$(Q,k)$ is in \FPT if on an input of size~$n$ it admits a solution running in time~\(f(k)\, n^{\Oh(1)}\), for some computable function~$f$, and note that the polynomial on~$n$ does not depend on~$k$. We refer the reader to the textbook of~\citet{cygan_parameterized_2015} for a detailed reference of the subject.

\subsubsection{Parameterized hardness}\label{subsec:param-hard}

The analogue for \NP-hard problems in parameterized complexity are $\W{t}$-hard problems, that is, a problem is $\W{t}$-hard if any problem in~$\W{t}$ has a parameterized reduction to it and, while~$\FPT \subseteq \W{t}$, it is believed that~$\FPT \neq \W{t}$ for any~\(t \geq 1\). We now present the formal definition of~$\W{t}$. 

A \emph{Boolean circuit} is an acyclic directed graph with \emph{input nodes} (of indegree zero), \emph{negation nodes} (of indegree one), and \emph{and-nodes} and \emph{or-nodes} (of indegree at least two).
Additionally, exactly one node is an \emph{output node}.
A \emph{satisfying assignment} is an assignment of $0/1$ values to the input nodes such that the circuit computes (in the obvious manner) the value $1$ in the output node. 
The \emph{weight} of an assignment is the number of input nodes that are assigned value~$1$.

\paraProblemDef{\pname{Weighted circuit satisfiability}}{Boolean circuit $C$ and positive integer $k$.}{$k$}{Does $C$ admit a satisfying assignment of weight $k$?}{pro:weightsat}

The \emph{depth} of the circuit is the longest path from an input node to the output node.
\emph{Large nodes} in the circuit are nodes with indegree greater than two.
The \emph{weft} of the circuit is the maximum number of large nodes in a path from an input node to the output node.

By definition, a problem $\pi$ is in $\W{t}$ if there is parameterized reduction from $\pi$ to \nameref{pro:weightsat} restricted to circuits of constant depth and weft at most~$t$. It is clear that~\(\W{t} \subseteq \W{t+1}\).

Deciding if a graph has a clique of size~$k$, parameterized by~$k$, is the classic~\W1-complete problem, while deciding if a graph has dominating set of size at most~$k$, parameterized by~$k$, is the classic~\W2-complete problem.


A dominating set of an orientation~$D$ is a set~$S \subseteq V(D)$ such that each vertex of~$D$ either is in~$S$ or is the head of an arc outgoing from a vertex in~$S$. The \emph{domination number} of~$D$, denoted~\(\odn(D)\), is the smallest size of a dominating set of~$D$. Consider the following parameterized problem.

\paraProblemDef{\pname{Parameterized Dominating Set on Tournaments}}{Tournament graph $T$ and positive integer $k$.}{$k$}{Is $\odn(G)\leq k$?}{pro:TDSparam}

\citet{cygan_parameterized_2015} shows that~\nameref{pro:TDSparam} is \W2-complete.
Moreover, the following result follows directly from their Observation 14.22 and Corollary 14.23, and we use it to obtain strong lower bounds for our problems.

\begin{lemma}[\citet{cygan_parameterized_2015}]\label{res:eth-param}
    If there exists a parameterized reduction from~\nameref{pro:TDSparam} to a parameterized problem~$(Q,k)$, then~$(Q,k)$ is \W2-hard. Moreover, if the reduction has a linear parameter dependence then, for every computable~$f$, there is no~$f(k)\,n^{o(k)}$ time algorithm for~$(Q,k)$, unless~\ETH is falsified.
\end{lemma}

\begin{lemma}\label{res:eth-fpt-neq-w1}
    \(\FPT \neq \W1\), unless~\ETH is falsified.
\end{lemma}
\begin{proof}
    \ETH implies that~\textsc{Clique} is not in \FPT by~\cite[Theorem 14.21]{cygan_parameterized_2015}. Moreover,~\textsc{Clique} is~\W1-complete by~\cite[Theorem 13.25]{cygan_parameterized_2015}. Thus,~\ETH implies~$\FPT \neq \W1$.\qed 
\end{proof}

\section{Results for \texorpdfstring{$\opc$}{P3}-convexity}\label{sec:opc}

We formally define the discussed problems below, and provide upper and lower bounds for its classical and parameterized variants.

\problemDef{\pname{$\opc$-interval on Tournaments}}{Tournament graph $T$ and positive integer $k$.}{Is $\oinp(T)\leq k$?}{pro:tinp}

\paraProblemDef{\pname{Parameterized $\opc$-interval on Tournaments}}{Tournament graph $T$ and positive integer $k$.}{$k$}{Is $\oinp(T) \leq k$?}{pro:tinpparam}

\subsection{Classical complexity upper bounds}\label{subsec:oinp-alg}

We show that \nameref{pro:tinp} admits a quasi-polynomial algorithm.
This follows from the fact that the size of the $\opc$-interval set of a tournament $T$ is at most logarithmic on the number of vertices of $T$.
That fact has been proven by Medeiros~\cite{medeiros_convexidade_2024} in his PhD thesis.
\begin{restatable}{lemma}{inpLogBound}\label{res:inp-log-bound}%
    If \(T\) is a tournament with \(n \geq 2\) vertices, then \(\oinp(T) \leq \lceil 2\log_{\frac43} n \rceil.\)
\end{restatable}%
\begin{proof}
    We proceed by induction on~\(n = |V(T)|\).
    Since~\(\oinp(T) \leq n \leq \left\lceil 2\log_{4/3} n \right\rceil\) for all \(2 \leq n \leq 20\), the bound holds trivially in that case.

    Now assume the result holds for every tournament on~\(m\) vertices with~\(2 \leq m < n\).
    Since~\(T\) is a tournament, \(\sum_{v \in V(T)} d^+_T(v) = \binom{n}{2}\), and there exists~\(v_1 \in V(T)\) such that \(d^+_T(v_1) \geq \lceil \nicefrac{(n-1)}{2} \rceil\).
    Let~\(T_1 = T\bigl[N^+_T(v_1)\bigr]\), and by applying the same averaging argument to the indegrees of the vertices inside~\(T_1\), there exists~\(v_2 \in V(T_1)\) such that
        \(  d^-_{T_1}(v_2) \geq \left\lceil \nicefrac{(|V(T_1)| - 1)}{2} \right\rceil \geq \left\lceil \nicefrac{(n-3)}{4} \right\rceil. \)
    For each~\(w \in N^-_{T_1}(v_2)\), we have~\(v_1 w\) and~\(w v_2\) in~\(T\), so~\(v_1 w v_2\) is a oriented path of length two, and
     \( \{v_1, v_2\} \cup N^-_{T_1}(v_2) \subseteq \oifp\!\bigl(\{v_1, v_2\}\bigr). \)
    Therefore,
    \[ \bigl|\oifp\!\bigl(\{v_1,v_2\}\bigr)\bigr|
        \geq d^-_{T_1}(v_2) + 2
        \geq \left\lceil \frac{n-3}{4} \right\rceil + 2
        \geq \frac{n}{4}.\]
    Let~\(T' = T\bigl[V(T) \setminus \oifp(\{v_1,v_2\})\bigr]\), and note that~\(|V(T')| \leq 3n/4\). 
    If, $|V(T')| \leq 1$, then $\oifp(\{v_1,v_2\}) \geq n - 1$, which implies that $\oinp(T) \leq 3$.
    Thus, assume that $|V(T')| \geq 2$.
    By the inductive hypothesis, \(T'\) has an \(\opc\)-interval set \(S'\) of size at most \(\left\lceil 2\log_{\nicefrac43}|V(T')| \right\rceil \le \left\lceil 2\log_{\nicefrac43}(\nicefrac{3n}{4}) \right\rceil\).
    Note that \(S = \{v_1, v_2\} \cup S'\) is a \(\opc\)-interval set of \(T\), and thus
    \[
                        \oinp(T) \leq |S|
                \leq 2 + \left\lceil 2\log_{\frac{4}{3}}\frac{3n}{4}
                    \right\rceil 
                = 2 + \left\lceil 2\log_{\frac{4}{3}} n - 2 \right\rceil 
                = \left\lceil 2\log_{\frac{4}{3}} n \right\rceil.
    \]
    \qed
\end{proof}

\begin{restatable}{theorem}{tinpQuasiPoly}\label{res:tinp-quasi-poly}%
    \nameref{pro:tinp} can be decided in time \(n^{\Oh(\log n)}\).
\end{restatable}%
\begin{proof}
    Let~\((T, k)\) be an instance of~\nameref{pro:tinp}, where~\(T\) is a tournament on \(n\) vertices. By~\Cref{res:inp-log-bound}, if~$k > \lceil 2 \log_{\nicefrac43} n \rceil$ then the answer is trivially ``yes''. Otherwise~$k=\Oh(\log n)$, and there exists an algorithm that enumerates all subsets \(S \subseteq V(T)\) with \(|S| \leq k\) in~\(n^{\Oh(k)}\) time. Checking whether each~\(S\) is an \(\opc\)-interval set requires verifying, for every~\(v \in V(T) \setminus S\), the existence of~\(u, w \in S\) such that~\(uv,vw \in A(T)\), and this takes~$\Oh(n^2)$ time.
    Therefore, the total running time is \(n^{\Oh(k)}\cdot \Oh(n^2) = n^{\Oh(\log n)}\). \qed
\end{proof}

\subsection{Lower bounds}\label{subsec:oinp-w2-hard}

We show that deciding whether~\(\oinp(T) \leq k\) is unlikely to be in~\P, and it is hard even when parameterized by~\(k\). In fact, we argue the problem is \NP-intermediate, and thus the~$n^{\Oh(\log n)}$ time algorithm from the previous section is likely to be optimal.

\begin{theorem} \label{thm:tinpparam-w2-hard}
    There exists a linear-time many-one parameterized reduction from an instance~$(T,k)$ of \nameref{pro:TDSparam} to an instance~$(T', k+1)$ of~\nameref{pro:tinpparam}.
    
    That is, given a tournament~$T$, in~$\Oh(|A(T)|)$ time we can build a tournament~$T'$ such that~$\oinp(T') = \odn(T) + 1$.
\end{theorem}

\begin{proof}
    Let~\((T,k)\) be an instance of \nameref{pro:TDSparam}. Define~\(T'\) as the tournament obtained from~\(T\) by adding a single sink vertex~\(s\), that is, \(V(T') = V(T) \cup \{s\}\) and~\(A(T') = A(T) \cup \{us : u \in V(T)\}\). This construction clearly runs in linear time.

    Let~$S$ be a dominating set of~$T$. Since \(s\) is a sink, for each~$v \in V(T) \setminus S$ it follows that~$uvs$ is a path of length two in~$T'$, for some~$u \in S$. Thus \(S \cup \{s\}\) is a~\(\opc\)-interval set of~\(T'\) and \(\oinp(T') \leq \odn(T) +1\).

    Now let~\(S'\) be a~\(\opc\)-interval set of~$T'$. For each~$v \in V(T) \setminus S'$, there exist~$u, w \in S'$ such that~$uvw$ is path, and note~$u \neq s$ since~$s$ is a sink. Thus~$S' \setminus \{s\}$ is a dominating set of~$T$ and~\(\odn(T) \leq \oinp(T') - 1\). We conclude that~\(\oinp(T') = \odn(T) + 1\).\qed
\end{proof}

By~\Cref{res:eth-param}, we immediately get the following result.

\begin{corollary} \label{res:tinpparam-w2-hard}
    \nameref{pro:tinpparam} is \W2-hard. Moreover, there is no~$f(k)\,|V(T)|^{o(k)}$ time algorithm for it, for every computable~$f$, unless \ETH is falsified.
\end{corollary}

By combining our results, we notice that \nameref{pro:tinp} actually has an ``intermediate'' difficulty between having a polynomial solution and being~\NP-complete, similar to the problems of graph isomorphism and integer factorization.

\begin{corollary}
    \nameref{pro:tinp} is \NP-intermediate, unless~\ETH is falsified.
\end{corollary}

\begin{proof}
    Clearly, \nameref{pro:tinp} is in \NP{}, as a candidate set $S \subseteq V(T)$ can be verified in $\Oh(n^2)$ time. 

    By \Cref{res:tinpparam-w2-hard}, \nameref{pro:tinpparam} is \W2-hard, so \nameref{pro:tinp} cannot be in \P{}, unless $\FPT = \W2$, which also falsifies \ETH by~\Cref{res:eth-fpt-neq-w1}.
    
    Now suppose \nameref{pro:tinp} is \NP-complete. Then there is a reduction from any size-$n$ instance of~\textsc{3-SAT} to an $n^{\Oh(1)}$-size instance of~\nameref{pro:tinp}. Then, by~\Cref{res:tinp-quasi-poly}, \textsc{3-SAT} can be solved in quasi-polynomial $n^{\Oh(\log n)}$ time, which falsifies \ETH. \qed
\end{proof}

\subsection{Parameterized upper bounds}\label{subsec:oinp-w2}

We have proved a lower bound for \nameref{pro:tinpparam} in the \mbox{\cclass{W}-hierarchy}.
Now, we will show an upper bound: this problem is in \W2, concluding that it is in fact~\W2-complete. Refer to~\Cref{subsec:param-hard} for the definitions of the~\cclass{W}-hierarchy.
It suffices to show this upper bound for a more general version of the problem, not restricted to tournaments.

\paraProblemDef{\pname{Parameterized $\opc$-interval}}{Orientation $D$ and positive integer $k$.}{$k$}{Is $\oinp(D) \leq k$?}{pro:inpparam}

\begin{theorem}\label{res:inpparam-w2}
    \nameref{pro:inpparam} is in $\W2$.
\end{theorem}
\begin{proof}
    Let $(D, k)$ be an instance of \nameref{pro:inpparam}.
    We will show a reduction from $(D, k)$ to an instance $(C_D, k)$ of \nameref{pro:weightsat}, where $C_D$ is a Boolean circuit of weft two.

    We begin with the description of~$C_D$.
    For every vertex $v$ of $D$, create an input node $v_{\text{in}}$, as well as two large or-nodes $v_{-}$ and $v_{+}$, which receive arcs from all input nodes of~$N^-(v)$ and~$N^+(v)$, respectively. Create an and-node $v_1$ receiving arcs from~$v_{-}$ and~$v_{+}$, and an or-node $v_2$ with arcs from $v_{\text{in}}$ and $v_1$.
    Complete the circuit by creating a large output and-node~$t$ with arcs from all nodes~$v_2$.
    See Figure~\ref{fig:inp3-w2-gadget} for an illustration of~$C_D$.

\begin{figure}
    \centering
    \includegraphics[scale=0.6]{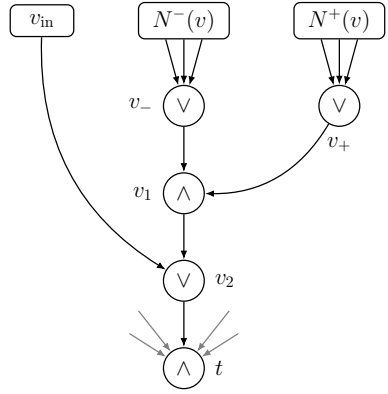}
    \caption{The circuit constructed in~\Cref{res:inpparam-w2}.}
    \label{fig:inp3-w2-gadget}
\end{figure}

    More formally, take
    \begin{align*}
        &V(C_D) =& \{t\}    && \cup & \bigcup\limits_{v \in V(D)} \{v_{\text{in}}, v_{-}, v_{+}, v_1, v_2\} \text{ and}  \\
        &A(C_D) =& \bigcup\limits_{uv \in A(D)} \{u_{\text{in}} v_{-}, v_{\text{in}} u_{+}\} && \cup & \bigcup\limits_{v \in V(D)} \{v_{-} v_1, v_{+} v_1, v_1 v_2, v_{\text{in}} v_2, v_2 t \}.
    \end{align*}

    Consider a satisfying assignment to~$C_D$, and let~$S$ be the set of all vertices~$v$ such that~$v_{\text{in}}$ is assigned~1.
    Notice that $v_1$ outputs 1 if and only if there exists a vertex $u \in N^-(v) \cap S$ and a vertex $w \in N^+(v) \cap S$.
    In that case, by definition, $v \in \oifp(\{u, w\})$.
    Therefore, $v_2$ outputs 1 if and only if $v \in \oifp(S)$.
    Clearly, the output node $t$ has value 1 if and only if $V(D) = \oifp(S)$, and thus~\(\oinp(D) \leq k\).
    
    We can build this circuit in $\Oh(|A(D)|)$ time; thus, we have a polynomial time reduction from $(D, k)$ to $(C_D, k)$.
    Moreover, notice that any path in the circuit~$C_D$ goes through at most two large nodes, namely~$t$ and at most one node of the form $v_{-}$ or $v_{+}$.
    Therefore, $C_D$ has weft two and the theorem follows.\qed
\end{proof}

\section{Results for \texorpdfstring{$\opsc$}{P3*}-convexity}\label{sec:opsc}

We formally define the problem discussed in the introduction for $\opsc$-convexity, and provide upper and lower bounds for its classical and parameterized variants.

\problemDef{\pname{$\opsc$-interval on Tournaments}}{Tournament graph $T$ and positive integer $k$.}{Is $\oinps(T)\leq k$?}{pro:tinps}

\paraProblemDef{\pname{Parameterized $\opsc$-interval on Tournaments}}{Tournament graph $T$ and positive integer $k$.}{$k$}{Is $\oinps(T) \leq k$?}{pro:tinpsparam}

\subsection{Family of non-trivial instances with large interval sets}\label{subsec:oinps-n}

One might wonder if the size of a minimum~\(\opsc\)-interval set is always small and thus a quasi-polynomial algorithm like the one from~\Cref{res:tinp-quasi-poly} exists for~\nameref{pro:tinps}. That is in fact not the case, as even in a transitive tournament, which is an acyclic tournament isomorphic to a total order, all vertices must be in any~\(\opsc\)-interval set.

A vertex~$v$ in a tournament~$T$ is called \emph{transitive} if~$uv, vw \in A(T)$ implies that~$uw \in A(T)$. Tournaments with transitive vertices are, in a sense, easy to deal with, since it is known that all transitive vertices are contained in any~\(\opsc\)-interval set~\cite[Proposition 1]{araujo_hull_2026}. 
In particular, strongly connected tournaments do not have transitive vertices.
In fact, we may assume that the tournament is strongly connected due to an efficient reduction to strongly connected components~\cite[Proposition 9]{araujo_hull_2026}. A \emph{strong} tournament is a tournament that is strongly connected.
The question of whether strong tournaments can have large minimum interval sets remained. We provide an infinite family of such tournaments below.

\Cref{fig:inps-example} shows the construction of the tournament family.
It was presented by Maia~\cite{maia_wbc_26} in the~\emph{2º Workshop Brasileiro de Combinatória} for a related problem.
\begin{restatable}{proposition}{oinpsLinear}\label{res:oinps-linear}%
For each~\(k \in \mathbb{N}\), there exists a strong tournament~\(T\) with~\(4k\) vertices such that~\(\oinps(T) = \Theta(k)\).
\end{restatable}%
\begin{proof}
We construct tournament~\(T\) with blocks~$B_i = \{u_1^i, u_2^i, u_3^i, u_4^i\}$ of four vertices, for each~$i \in [k]$. On each block, all arcs go left to right except for~$u_4^iu_1^i$. Between blocks, all arcs also go left to right, except for arcs~$u^{i+1}_1 u^i_4$, for each~$i \in [k-1]$. It is easy to see that tournament~$T$ is strong. See~\Cref{fig:inps-example} for an illustration of~\(T\).
\begin{figure}
    \centering
    \includegraphics[scale=0.8]{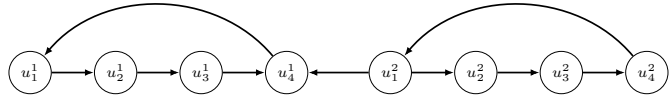}
    \caption{Example construction for $k=2$. All arcs not shown in the figure are oriented from left to right.}
    \label{fig:inps-example}
\end{figure}%

First note that~\(\oinps(T) \leq |V(T)| = \Oh(k)\).
Now let~\(S\) be an interval set of \(T\), and we claim that any \(\opsc\)-interval set~\(S\) of~\(T\) satisfies \(|S \cap B_i| \geq 2\) for every~\(i \in [k]\).
First, suppose that~\(u^i_2 \notin S\) for some~\(i \in [k]\).
Since~\(S\) is an interval set, there must exist~\(a, b \in S\) such that~\(a  u_2^i  b\) is an oriented minimum path of length two in~\(T\), that is,~\(ab \notin A(T)\).

We consider the possible locations of~\(a\) and~\(b\).
To satisfy~\(a u_2^i \in A(T)\), either~\(a = u_1^i \in B_i\), or~\(a \in B_\ell\) for some~\(\ell < i\).
Similarly, to satisfy~\(u_2^i  b \in A(S)\), either~\(b \in \{u_3^i, u_4^i\} \subseteq B_i\), or~\(b \in B_r\) for some~\(r > i\).
Note that any vertex in~$B_\ell$ for some~$\ell < i$ has outgoing arcs to all vertices in~\(\{u^i_3, u^i_4\} \cup B_r\) for any~$r > i$, thus since~$ab \notin A(T)$ it must be the case that~\(a = u^i_2\). Analogously, any vertex in~\(B_r\) for some~\(r > i\) has incoming arcs from all vertices in~\(\{u^i_1\} \cup B_\ell\) for any~\(\ell < i\), and thus~\(b \in \{u^i_3, u^i_4\}\).

Therefore both~\(a\) and~\(b\) must lie in~\(B_i\), and that the claim holds, that is, \(|S \cap B_i| \geq 2\). The case when~\(u^i_3 \notin S\) follows in a symmetrical fashion. Otherwise,~\(\{u^i_2, u^i_3\} \subseteq S\) and obviously the claim holds.
Therefore,~\(\oinps(T) \geq 2k\), and together with the upper bound it follows that~\(\oinps(T) = \Theta(n)\). \qed
\end{proof}
This implies that a brute force algorithm for \nameref{pro:tinps} which lists and tests all subsets of vertices in increasing size, may take~$2^{\Omega(n)}$ time, even when restricted to strong tournaments.

\subsection{Lower bounds}\label{subsec:oinps-w2-hard}

We show that deciding whether~\(\oinps(T) \leq k\) is unlikely to be in~\P, and it is hard even when parameterized by~\(k\). Therefore, the trivial~$n^{\Oh(k)}$ time algorithm is likely to be optimal.

Note that, in a tournament~\(T\), if~$uvw$ is a \emph{minimum} oriented path of size two, then~$wu \in A(T)$ and we say~\(uvw\) is an oriented triangle, which yields the following.

\begin{lemma} \label{res:int-triangle}
    If~$T$ is a tournament, then~$S \subseteq V(T)$ is a $\opsc$-interval set of~$T$ if and only if, for every~$v \in V(T) \setminus S$, there exist~$u, w \in S$ such that~$uvw$ is an oriented triangle.
\end{lemma}

\begin{theorem} \label{res:domt-inps-reduction}
    There exists a linear-time many-one parameterized reduction from an instance~$(T,k)$ of \nameref{pro:TDSparam} to an instance~$(T', k+2)$ of~\nameref{pro:tinpsparam}.
    
    More precisely, given a tournament~$T$, in~$\Oh(|A(T)|)$ time we can build a tournament~$T'$ such that~$\oinps(T') = \odn(T) + 2$.
\end{theorem}

\begin{proof}
    Given a tournament~$T$, build~$T'$ starting with two copies~$T_1$ and~$T_2$ of~$T$, and add to them an oriented triangle~$xyz$. For each arc~$uv \in A(T)$, add arcs~$u_1 v_2$ and~$u_2 v_1$, as well as arcs~$u_1 u_2$ for each~$u \in V(T)$. Finally, add arcs from~$x$ to every vertex in~$T_1$, and all other arcs are incoming to the triangle~$xyz$. See~\Cref{fig:domt-inpst} for an illustration of the reduction.
\begin{figure}
    \centering
    \includegraphics[scale=0.5]{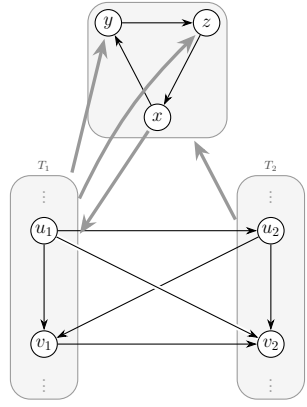}
    \caption{Construction of the tournament $T'$ in the reduction from~\Cref{res:domt-inps-reduction}. The tournaments $T_1$ and $T_2$ are identical copies of the original tournament $T$, and~$uv \in A(T)$. Thick gray arrows represent all-to-all arcs; for example, the arrow from $T_2$ to the $xyz$ bounding box indicates that every vertex in $T_2$ has an outgoing arc to $x$, $y$, and $z$.}
    \label{fig:domt-inpst}
\end{figure}

    More formally, take
    \begin{alignat*}{2}
        V(T') &= \{x,y,z\}    && \cup \bigcup\limits_{u \in V(T)} \{u_1,u_2\} \text{ and}  \\
        A(T') &= \{xy,yz,zx\} && \cup \bigcup\limits_{uv \in A(T)} \{u_1 v_1, u_2 v_2, u_1 v_2, u_2 v_1\} \\
              &               && \cup \bigcup\limits_{u \in V(T)} \{u_1 u_2, x u_1, u_1 y, u_1 z, u_2x, u_2y, u_2z\}.
    \end{alignat*}
    
    The tournament~$T'$ can be easily constructed in~$\Oh(|A(T)|)$ time.     
    We now argue that~$\oinps(T') = \odn(T) + 2$. Let~$S \subseteq V(T)$ be a dominating set of~$T$, and take~$S' = \{u_1 \mid u \in S\} \cup \{x,z\}$. Then, for each~$v \in V(T)$,
    \begin{enumerate}[label=(\roman*)]
        \item $x v_1 z$ is an oriented triangle, with $x,z \in S'$;
        \item if~$v \in S$, then~$v_1 v_2 x$ is an oriented triangle with~$v_1, x \in S'$;
        \item if~$v \notin S$, then there exists~$u \in S$ such that~$uv \in A(T)$, thus~$u_1 v_2 x$ is an oriented triangle with~$u_1, x \in S'$.
    \end{enumerate}
    Finally,~$xyz$ is an oriented triangle with~$x,z \in S'$. Therefore, $S'$ is a~$\opsc$-interval set of~$T'$ by~\Cref{res:int-triangle}, which proves~$\oinps(T') \leq \odn(T) + 2$.

    For the other direction, let~$S'$ be a~$\opsc$-interval set of~$T'$. First, suppose that~$|S' \cap \{x,y,z\}| \leq 1$. The only oriented triangle containing~$y$ is~$xyz$, meaning if~$y \notin S'$ then~$x,z \in S'$ by~\Cref{res:int-triangle}. So~$\{y\} = S' \cap \{x, y, z\}$. However,~$zx$ is the only outgoing arc from~$z$, thus all oriented triangles containing~$z$ also contain~$x$. This contradicts~\Cref{res:int-triangle}, because $x \notin S'$, but $z$ is also not in $S'$.
    Thus~$S'$ has at least two elements from~$\{x, y, z\}$.

    Now take~$S = \{u \mid u_1 \in S' \text{ or } u_2 \in S'\}$, and note that we just showed~$|S| \leq |S'| - 2$. Let~$v \in V(T) \setminus S$, hence~$v_1, v_2 \notin S'$. Since~$\oifps(S') = V(T')$, there exists~$u \in S'$ such that~$uv_2 \in A(T')$. Note~$v_2$ has no incoming arc from~$\{x,y,z\}$ by the construction of~$T'$. Thus either~$u = w_1$ or~$u = w_2$ for some~$w \in V(T) \setminus\{v\}$. In either case~$wv \in A(T)$ by the construction of~$T'$, and~$w \in S$ by the construction of~$S$. Therefore,~$S$ is a dominating set, which proves~$\odn(T) \leq \oinps(T') - 2$.\qed
\end{proof}

By~\Cref{res:eth-param}, we immediately get the following.

\begin{corollary} \label{res:inps-w2-hard}
    \nameref{pro:tinpsparam} is \W2-hard. Moreover, there is no~$f(k)\,|V(T)|^{o(k)}$ time algorithm for it, for every computable~$f$, unless \ETH is falsified.
\end{corollary}

The consequences for a polynomial algorithm for \nameref{pro:tinps} are more clearly spelled out in the following corollary.

\begin{corollary} \label{res:inps-not-p}
    \nameref{pro:tinps} is not in~\P, unless~$\FPT = \W2$ or \ETH is falsified.
\end{corollary}

\begin{proof}
    Suppose \nameref{pro:tinps} is in \P. Then, every parameterization of it is in~\FPT, including~\nameref{pro:tinpsparam}. Therefore, by~\Cref{res:inps-w2-hard}, all problems in~\W2 are in~\FPT, thus~$\FPT = \W2$. By~\Cref{res:eth-fpt-neq-w1}, this also means~\ETH is falsified. \qed
\end{proof}

\Cref{res:inps-not-p} is slightly weaker than stating~\nameref{pro:tinps} is~\NP-hard.\footnote{That is because~$\FPT \neq \W2$ implies~$\P \neq \NP$, but not the other way around. However, both conjectures are still widely believed to be true.} Also note that, while~\ETH implies~\pname{Dominating Set on Tournaments}~is not~\NP-hard, that is not the case for~\nameref{pro:tinps}. In fact, we conjecture it \emph{is}~\NP-hard.

\subsection{Parameterized upper bounds}\label{subsec:oinps-w2}

We now prove that \nameref{pro:tinpsparam} is \W2-complete.
As in the case of $\opc$-interval sets, it suffices to show that a more general version of the problem, not restricted to tournaments, is in \W2.

\paraProblemDef{\pname{Parameterized $\opsc$-interval}}{Orientation $D$ and positive integer $k$.}{$k$}{Is $\oinps(D) \leq k$?}{pro:inpsparam}

\begin{theorem}\label{res:inpsparam-w2}
    \nameref{pro:inpsparam} is in $\W2$.
\end{theorem}
\begin{proof}
    Let $(D, k)$ be an instance of \nameref{pro:inpsparam}.
    We will show a reduction from $(D, k)$ to an instance $(C_D, k)$ of \nameref{pro:weightsat}, where $C_D$ is a Boolean circuit of weft two.

    We begin with the description of~\(C_D\).
    For every vertex $v$ of $D$, create an input node $v_{\text{in}}$.
    For every pair of vertices $ab \subseteq N^-(v) \times N^+(v)$ where there is no arc~\(ab\), create an \mbox{and-node} $v_{ab}$ with arcs incoming from~\(a_{\text{in}}\) and~\(b_{\text{in}}\).
    Create a large or-node $v_1$ with arcs incoming from every~$v_{ab}$, as well as from~\(v_{\text{in}}\).
    Complete the circuit by creating a large output and-node~\(t\) with arcs incoming from~\(v_1\) for each~\(v \in V(D)\).
    See Figure~\ref{fig:inps-w2-gadget} for an illustration of~\(C_D\).

    \begin{figure}
        \centering
        \includegraphics[scale=0.6]{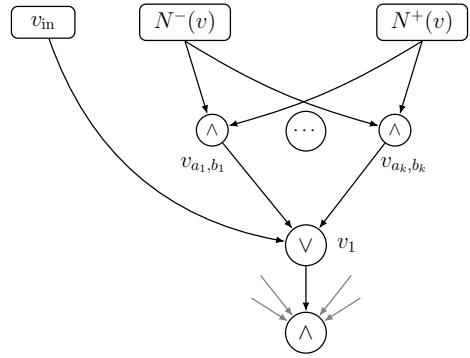}
        \caption{The circuit constructed in~\Cref{res:inpsparam-w2}.}
        \label{fig:inps-w2-gadget}
    \end{figure}

    More formally, let~\(N'_v = \{ab \mid av, vb \in A(D) \wedge ab \notin A(D)\}\) and take
    \begin{align*}
        V(C_D) &= \{t\}  \cup  \bigcup\limits_{v \in V(D)} \left(\{v_{\text{in}}, v_1\} \cup \{v_{ab} \mid ab \in N'_v \} \right)  \text{ and}  \\
        A(C_D) &= \bigcup\limits_{v \in V(D)}{\left(\{v_{\text{in}} v_1, v_1 t\} \cup \{a_{\text{in}} v_{ab}, b_{\text{in}} v_{ab}, v_{ab} v_1  \mid ab \in N'_v\} \right)}.
    \end{align*}
    
    Consider a satisfying assignment of weight~\(k\) to~\(C_D\), and let~\(S\) be the set of all vertices~\(v\) such that~\(v_{\text{in}}\) is assigned~1.
    Clearly, if $v \in S$, $v_1$ outputs 1; then suppose that $v \notin S$.
    In that case, $v_1$ outputs~1 if and only if~$v_{ab}$ outputs~1 for some~$a,b$, that is, if $a \in N^-(v) \cap S$, $b \in N^+(v) \cap S$, and~\(ab \notin A(D)\).
    In other words, $v$ is in the shortest path of length two from $a$ to $b$, thus $v \in \oifps(\{a, b\}) \subseteq \oifps(S)$.
    Therefore, $v_1$ outputs 1 if and only if $v \in \oifps(S)$, and the output node~\(t\) has value~1 if and only if $V(D) = \oifps(S)$, in which case~\(\oinps(D) \leq k\).
    
    We have a polynomial time reduction from $(D, k)$ to $(C_D, k)$, and~$C_D$ has weft two, since each path from an input to the output goes through at most one large node of the form~\(v_{ab}\), plus the output~\(t\) itself.
    Therefore, the theorem follows.\qed  
\end{proof}

\section{Conclusion and further research}\label{sec:conclusion}

We show that \nameref{pro:tinpparam} is \W2-hard and, under \ETH, its classical decision version \nameref{pro:tinp} is \NP-intermediate. In contrast, for \nameref{pro:tinpsparam}, we only establish \W2-hardness. We present a family of non-trivial instances with~\(\oinps(T) = \Theta(|V(T)|)\), and argue it is unlikely to have a subexponential solution, which we believe is the case.
\begin{conjecture}
    \nameref{pro:tinps} is \NP-complete.
\end{conjecture}

For further research, there are interesting variants of the parameters and convexities we presented.
For the geodesic convexity (\(\ogc\)), in which we consider any minimum path, not just those of length two, one can consider the complexity of deciding whether~\(\oing(T) \leq k\) when~\(T\) is a strong tournament, which was also raised as an open question by~\citet[Problem 31]{araujo_hull_2026}. Moreover, one can explore the complexity of deciding whether~$\ohnps(T) \leq k$ or~\(\ohng(T) \leq k\) in oriented split graphs and oriented cobipartite graphs, which is also still open.

\subsubsection{Acknowledgments}

We obtained the initial results which led to this paper in the \textit{2º Workshop Brasileiro de Combinatória} (WBC 2026).
We are grateful to the event organizers and, in particular, to Ana Karolina Maia and Júlio Araújo who introduced us to this very interesting topic and its open problems.

This study was financed, in part, by the São Paulo Research Foundation (FAPESP), Brasil. Process Numbers \#2026/00744-0 (to I. Capozzoli), \#2024/18049-0 (to Y. Couto) and \#2026/00996-9 (to E. Junchaya).

\bibliography{references}


\end{document}